\documentclass[11pt,onecolumn,draftcls]{IEEEtran}

\usepackage{amsmath, amssymb, multirow, algorithmic, algorithm, tikz}
\usepackage[caption=false]{subfig}

\newtheorem{definition}{Definition}
\newtheorem{proposition}{Proposition}
\newtheorem{theorem}{Theorem}
\newtheorem{example}{Example}
\newtheorem{lemma}{Lemma}
\newtheorem{corollary}{Corollary}

\newcommand{\rk}{\mathrm{rk}}
\newcommand{\cl}{\mathrm{cl}}
\newcommand{\mias}{\mathrm{mias}}
\newcommand{\vcup}{\,\vec{\cup}\,}
\newcommand{\bcup}{\,\bar{\cup}\,}

\begin{document}

\title{Closure solvability for network coding and secret sharing}
\author{Maximilien Gadouleau~\IEEEmembership{Member, IEEE}\thanks{The author is with the School of Engineering and Computing Sciences, Durham University, Durham, UK. Email: m.r.gadouleau@durham.ac.uk}}
\maketitle

\begin{abstract}
Network coding is a new technique to transmit data through a network by letting the intermediate nodes combine the packets they receive. Given a network, the network coding solvability problem decides whether all the packets requested by the destinations can be transmitted. In this paper, we introduce a new approach to this problem. We define a closure operator on a digraph closely related to the network coding instance and we show that the constraints for network coding can all be expressed according to that closure operator. Thus, a solution for the network coding problem is equivalent to a so-called solution of the closure operator. We can then define the closure solvability problem in general, which surprisingly reduces to finding secret-sharing matroids when the closure operator is a matroid. Based on this reformulation, we can easily prove that any multiple unicast where each node receives at least as many arcs as there are sources is solvable by linear functions. We also give an alternative proof that any nontrivial multiple unicast with two source-receiver pairs is always solvable over all sufficiently large alphabets.  Based on singular properties of the closure operator, we are able to generalise the way in which networks can be split into two distinct parts; we also provide a new way of identifying and removing useless nodes in a network. We also introduce the concept of network sharing, where one solvable network can be used to accommodate another solvable network coding instance. Finally, the guessing graph approach to network coding solvability is generalised to any closure operator, which yields bounds on the amount of information that can be transmitted through a network.
\end{abstract}

\section{Introduction} \label{sec:intro}

Network coding \cite{ACLY00} is a protocol which outperforms routing for multicast networks by letting the intermediate nodes manipulate the packets they receive. In particular, linear network coding \cite{LYC03} is optimal in the case of one source; however, it is not the case for multiple sources and destinations \cite{Rii04, DFZ05}. Although for large dynamic networks, good heuristics such as random linear network coding \cite{KM03, HMK+06} can be used, maximizing the amount of information that can be transmitted over a static network is fundamental but very difficult in practice. Solving this problem by brute force, i.e. considering all possible operations at all nodes, is computationally prohibitive. Different alternative approaches have been proposed to tackle this problem, notably using matroids, information inequalities, and group theory \cite{DFZ04, DFZ06, DFZ07, CG08, CG10, Cha05}. In this paper, we provide a new approach to tackle this problem based on a closure operator defined on a related digraph. Closure operators are fundamental and ubiquitous mathematical objects.

The guessing number of digraphs is a concept introduced in \cite{Rii06}, which connects graph theory, network coding, and circuit complexity theory. In \cite{Rii06} it was proved that an instance of network coding with $r$ sources and $r$ sinks on an acyclic network (referred to as a multiple unicast network) is solvable over a given alphabet if and only if the guessing number of a related digraph is equal to $r$. Moreover, it is proved in \cite{Rii06, DZ06} that any network coding instance can be reduced into a multiple unicast network. Therefore, the guessing number is a direct criterion on the solvability of network coding. One of the main advantages of the guessing number approach is to remove the hierarchy between sources, intermediate nodes, and destinations. In \cite{Rii07}, the guessing number is also used to disprove a long-standing open conjecture on circuit complexity. In \cite{WCR09}, the guessing number of digraphs was studied, and bounds on the guessing number of some particular digraphs were derived. The guessing number is also equal to the so-called graph entropy \cite{Rii06, Rii07b}. This allows us to use information inequalities \cite{Cha11} to derive upper bounds on the guessing number. The guessing number of undirected graphs is studied in \cite{CM11}. Moreover, in \cite{GR11}, the guessing number is viewed as the independence number of an undirected graph related to the digraph and the alphabet.


Shamir introduced the so-called threshold secret sharing scheme in \cite{Sha79}. Suppose a sender wants to communicate a secret $a \in A$ to $n$ parties, but that an eavesdropper may intercept $r-1$ of the transmitted messages. We then require that given any set of $r-1$ messages, the eavesdropper cannot obtain any information about the secret. On the other hand, any set of $r$ messages allows to reconstruct the original secret $a$. The elegant technique consists of sending evaluations of a polynomial $p(x) = \sum_{i=0}^{r-1} p_ix^i$, with $p_0 = a$ and all the other coefficients chosen secretly at random, at $n$ nonzero elements of $A$; this is evidently reminiscent of Reed-Solomon codes. The threshold scheme was then generalised to {\em ideal secret sharing schemes} with different access structures, i.e. different sets of trusted parties. Brickell and Davenport have proved that the access structure must be the family of spanning sets of a matroid; also any linearly representable matroid is a valid access structure \cite{BD91}. However, there exist matroids (such as the V\'amos matroid \cite{Sey92}) which are not valid access structures. For a given access structure (or equivalently, matroid), finding the scheme is equivalent to a representation by partitions \cite{Mat99}.

In this paper, we introduce a closure operator on digraphs, and define the closure solvability problem for any closure operator. This yields the following contributions.
\begin{itemize}
	\item First of all, this framework encompasses network coding and ideal secret sharing. In particular, network coding solvability is equivalent to the solvability of the closure operator of a digraph associated to the network. This framework then allows us to think of network coding solvability on a higher, more abstract level. The problem, which used to be about coding functions, is now a simplified problem about partitions. 
				
	\item This approach is particularly elegant, in different aspects. Firstly, the adjacency relations of the graph, and hence the topology of the network, are not visible in the closure operator. 
Therefore, the closure operator filters out some unnecessary information from the graph. Secondly, it is striking that all along the paper, most proofs will be elementary, including those of far-reaching results. Thirdly, this framework highlights the relationship with matroids unveiled in \cite{DFZ07, CGR13}.
	
	\item Like the guessing number approach, the closure operator approach also gets rid of the source-intermediate node-destination hierarchy. The guessing graph machinery of \cite{GR11} can then be easily generalised to any closure operator. In other words, the interesting aspects of the guessing number approach can all be recast and generalised in our framework.
	
	\item This approach then yields interesting results. First, it was shown in \cite{GR11} that the entropy of a digraph is equal to the sum of the entropies of its strongly connected components. Thus, one can split the solvability problem of a digraph into multiple ones, one for each strongly connected component \cite{GR11}. In this paper, we extend this way of splitting the problem by considering the closure operators induced by the subgraphs. We can easily exhibit a strongly connected digraph whose closure operator is disconnected, i.e. which can still be split into two smaller parts. More specifically, if the graph is strongly connected but its closure operator is disconnected, then we can exhibit a set of vertices which are simply useless and can be disregarded for solvability. Second, we can prove that any digraph whose closure operator has rank two is solvable. This means that any multiple unicast with two source-receiver pairs is solvable, unless there exists an easily spotted bottleneck in the network. This has already been proved in \cite{WS10}; our proof is much shorter and highlights the relation with coding theory and designs. Third, we can prove that any network with minimum in-degree equal to the number of source-receiver pairs is solvable by linear functions over all sufficiently large alphabets of size equal to a large prime power. Fourth, we prove an equivalence between network coding solvability and index coding solvability. Finally, we show how the bidirectional union of digraphs can be viewed as network sharing.
\end{itemize}

The rest of the paper is organised as follows. In Section \ref{sec:preliminaries}, we review some useful background. In Section \ref{sec:clD}, we define the closure solvability problem and prove that network coding solvability is equivalent to the solvability of a closure operator. We then use this conversion in Section \ref{sec:solvable} to prove the solvability of different classes of networks. We investigate how to combine closure operators in Section \ref{sec:combining}. We finally define the solvability graph in \ref{sec:solvability_graph} and study its properties.

\section{Preliminaries} \label{sec:preliminaries}

\subsection{Closure operators} \label{sec:closure}

Throughout this paper, $V$ is a set of $n$ elements. A closure operator on $V$ is a mapping $\cl: 2^V \to 2^V$ which satisfies the following properties \cite[Chapter IV]{Bir48}. For any $X,Y \subseteq V$,
\begin{enumerate}
    \item \label{it:Xincl} $X \subseteq \cl(X)$ (extensive);

    \item \label{it:cl_increasing} if $X \subseteq Y$, then $\cl(X) \subseteq \cl(Y)$ (isotone);

    \item \label{it:cl(cl)} $\cl(\cl(X)) = \cl(X)$ (idempotent).
\end{enumerate}
A {\em closed set} is a set equal to its closure. For instance, in a group one may define the closure of a set as the subgroup generated by the elements of the set; the family of closed sets is simply the family of all subgroups of the group. Another example is given by linear spaces, where the closure of a set of vectors is the subspace they span.

A closure operator satisfies the following properties. For any $X,Y \subseteq V$,
\begin{enumerate}
    \item \label{it:cl=bigcap} $\cl(X)$ is equal to the intersection of all closed sets containing $X$;

    \item \label{it:cl(cap)} $\cl(\cl(X) \cap \cl(Y)) = \cl(X) \cap \cl(Y)$, i.e. the family of closed sets is closed under intersection;

    \item \label{it:cl(cup)} $\cl(X \cup Y) = \cl(\cl(X) \cup \cl(Y))$.

    \item \label{it:axiom} $X \subseteq \cl(Y)$ if and only if $\cl(X) \subseteq \cl(Y).$
\end{enumerate}

We refer to
$$
    r:= \min\{|b| : \cl(b) = V\}
$$
as the {\em rank} of the closure operator. For instance, in a linear space, this is the dimension of the space. Any set $b \subseteq V$ of size $r$ and whose closure is $V$ is referred to as a {\em basis} of $\cl$.

An important class of closure operators are {\em matroids} \cite{Oxl06}, which satisfy the Mac Lane-Steinitz exchange property\footnote{In order to simplify notation, we shall identify a singleton $\{v\}$ with its element $v$}: if $X \subseteq V$, $v \in V$ and $u \in \cl(X \cup v) \backslash \cl(X)$, then $v \in \cl(X \cup u)$. A special class consists of the uniform matroids, typically denoted as $U_{r,n}$, where
$$
    U_{r,n}(X) = \begin{cases} V & \mbox{if } |X| \ge r\\
    X & \mbox{otherwise}. \end{cases}
$$
Clearly, $U_{r,n}$ has rank $r$.

\subsection{Functions and their kernels} \label{sec:functions}

While network coding typically works with functions assigned to vertices, it is elegant to work with {\em partitions} (for a review of their properties, the reader is invited to \cite{Bai04}). Recall that a partition of a set $B$ is a collection of subsets, called parts, which are pairwise disjoint and whose union is the whole of $B$. We denote the parts of a partition $f$ as $P_i(f)$ for all $i$. 

Any function $\bar{f}: B \to C$ has a {\em kernel} denoted as $f := \{\bar{f}^{-1}(c): c \in \bar{f}(B)\}$, defined by the partition of $B$ into pre-images under $\bar{f}$. Conversely, any partition of $B$ in at most $|C|$ can be viewed as the kernel of some function from $B$ to $C$. Note that two functions $\bar{f}$, $\bar{g}$ have the same kernel if and only if $\bar{f} = \pi \circ \bar{g}$ for some permutation $\pi$ of $C$. The kernel of any injective function $B \to C$ is the so-called equality partition $E_B$ of $B$ (i.e. the partition with $|B|$ parts). 

If any part of $f$ is contained in a unique part of $g$, we say $f$ refines $g$. The equality partition refines any other partition, while the universal partition (the partition with one part) is refined by any other partition. The {\em common refinement} of two partitions $f$, $g$ of $B$ is given by $h:= f \vee g$ with parts
$$
    P_{i,j}(h) = P_i(f) \cap P_j(g): P_i(f) \cap P_j(g) \ne \emptyset.
$$

We shall usually consider a tuple of $n$ partitions $f= (f_1,\ldots,f_n)$ assigned to elements of a finite set $V$ with $n$ elements. In that case, for any $X \subseteq V$, we denote the common refinement of all $f_v, v \in X$ as $f_X := \bigvee_{v \in X} f_v$. For any $S,T \subseteq V$ we then have $f_{S \cup T} = f_S \vee f_T$.

\subsection{Digraphs}

Throughout this paper, we shall only consider digraphs \cite{BG09a} with no repeated arcs. We shall denote the arc set as $E(D)$, since the letter $A$ will be reserved for the alphabet. However, we do allow edges in both directions between two vertices, referred to as {\em bidirectional edges} (we shall abuse notations and identify a bidirectional edge with a corresponding undirected edge) and loops over vertices. In other words, the digraphs considered here are of the form $D = (V,E)$, where $E \subseteq V^2$. For any vertex $v$ of $D$, its in-neighborhood is $v^- = \{u \in V: (u,v) \in E(D)\}$ and its in-degree is the size of its in-neighborhood. By extension, we denote $X^- = \bigcup_{v \in X} v^-$ for any set of vertices $X$. Also, by analogy, the out-neighbourhood of $v$ is $v^+ := \{u \in V: (v,u) \in E(D)\}$. We say that a digraph is {\em strongly connected} if there is a path from any vertex to any other vertex of the digraph. 

The {\em girth} of a digraph is the minimum length of a cycle, where we consider a bidirectional edge as a cycle of length $2$. A digraph is {\em acyclic} if it has no directed cycles. In this case, we can order the vertices $v_1,\ldots,v_n$ so that $(v_i,v_j) \in E(D)$ only if $i<j$. The cardinality of a maximum induced acyclic subgraph of the digraph $D$ is denoted as $\mias(D)$. A set of vertices $X$ is a {\em feedback vertex set} if and only if any directed cycle of $D$ intersects $X$, or equivalently if $V \backslash X$ induces an acyclic subgraph. 

\begin{definition} \label{def:D_union}
\cite{GR11} For any digraphs $D_1$ and $D_2$ with disjoint vertex sets $V_1$ and $V_2$, we denote the disjoint union, unidirectional union, and bidirectional union of $D_1$ and $D_2$ as the graphs on $V_1 \cup V_2$ and respective edge sets
\begin{align*}
	E(D_1 \cup D_2) &= E(D_1) \cup E(D_2)\\
	E(D_1 \vcup D_2) &= E(D_1 \cup D_2) \cup \{(v_1,v_2): v_1 \in V_1, v_2 \in V_2\}\\
	E(D_1 \bcup D_2) &= E(D_1 \vcup D_2) \cup \{(v_2,v_1): v_1 \in V_1, v_2 \in V_2\}.
\end{align*}
\end{definition}

In other words, the disjoint union simply places the two graphs next to each other; the unidirectional union adds all possible arcs from $D_1$ to $D_2$ only; the bidirectional union adds all possible arcs between $D_1$ and $D_2$.

\subsection{Guessing game and guessing number} \label{sec:preliminaries_guessing}

A \textit{configuration} on a digraph $D$ on $V$ over a finite alphabet $A$ is simply an $n$-tuple $x = (x_1, \ldots, x_n) \in A^n$. A \textit{protocol} $f = (f_1,\ldots,f_n)$ of $D$ is a mapping $f:A^n \to A^n$ such that $f(x)$ is locally defined, i.e. $f_v(x) = f_v(x_{v^-})$ for all $v$. The fixed configurations of $f$ are all the configurations $x \in A^n$ such that $f(x) = x$. The {\em guessing number} of $D$ is then defined as the logarithm of the maximum number of configurations fixed by a protocol of $D$:
\begin{equation} \label{eq:def_g} \nonumber
	g(D,A) = \max_f \left\{ \log_{|A|} |\mathrm{Fix}(f)| \right\}.
\end{equation}

We now review how to convert a multiple unicast problem in network coding to a guessing game. Note that any network coding instance can be converted into a multiple unicast without any loss of generality \cite{DZ06, Rii07}. We suppose that each sink requests an element from an alphabet $A$ from a corresponding source. This network coding instance is {\em solvable} over $A$ if all the demands of the sinks can be satisfied at the same time. We assume the network instance is given in its {\em circuit representation}, where each vertex represents a distinct coding function and hence the same message flows every edge coming out of the same vertex; again this loses no generality. This circuit representation has $r$ source nodes, $r$ sink nodes, and $m$ intermediate nodes. By merging each source with its corresponding sink node into one vertex, we form the digraph $D$ on $n = r+m$ vertices. In general, we have $g(D,A) \le r$ for all $A$ and the original network coding instance is solvable over $A$ if and only if $g(D,A) = r$ \cite{Rii07}. Note that the protocol on the digraph is equivalent to the coding and decoding functions on the original network.

For any digraphs $D_1$, $D_2$ on disjoint vertex sets $V_1$ and $V_2$ respectively, we have 
\begin{align*}
	g(D_1 \cup D_2,A) &= g(D_1 \vcup D_2,A) = g(D_1,A) + g(D_2,A),\\
	g(D_1 \bcup D_2, A) &\le \min\{|V_1| + g(D_2,A), |V_2| + g(D_1,A)\},
\end{align*}
for all alphabets $A$ \cite{GR11}. Notably, we can always consider strongly connected graphs only.

We illustrate the conversion of a network coding instance to a guessing game for the famous butterfly network in Figure \ref{fig:butterfly}. It is well-known that the butterfly network is solvable over all alphabets, and conversely it was shown that the clique $K_3$ has guessing number $2$ over any alphabet. The combinations and decoding operations on the network are equivalent to the protocol on the digraph. For instance, if $v_3$ transmits the opposite of the sum of the two incoming messages modulo $|A|$ on the network, the corresponding protocol lets all nodes guess minus the sum modulo $|A|$ of their incoming elements.

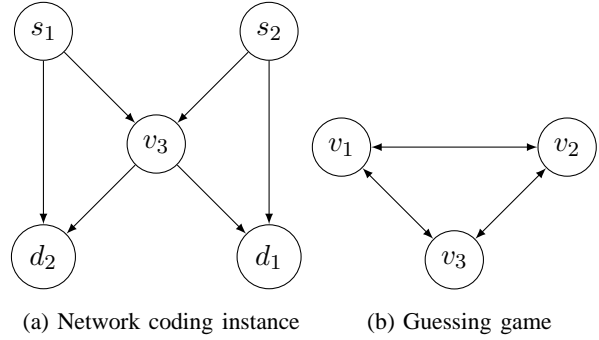
\begin{figure}
\centering
	\subfloat[Network coding instance]
	{\begin{tikzpicture}
	\tikzstyle{every node}=[draw,shape=circle];
	
	\node (xs) at (0,3) {$s_1$};
	\node (xr) at (3,0) {$d_1$};
	
	\node (ys) at (3,3) {$s_2$};
	\node (yr) at (0,0) {$d_2$};
	
	\node (z) at (1.5,1.5) {$v_3$};
	
	\draw[-latex] (xs) -- (yr);
	\draw[-latex] (ys) -- (xr);
	
	\draw[-latex] (xs) -- (z);
	\draw[-latex] (ys) -- (z);
	
	\draw[-latex] (z) -- (xr);
	\draw[-latex] (z) -- (yr);
	\end{tikzpicture}
	}
	\subfloat[Guessing game]
	{\begin{tikzpicture}
	\tikzstyle{every node}=[draw,shape=circle];
	
	\node (x) at (0,3) {$v_1$};
	\node (y) at (3,3) {$v_2$};
	\node (z) at (1.5,1.5) {$v_3$};
	
	\draw[latex-latex] (x) -- (y);
	\draw[latex-latex] (x) -- (z);
	\draw[latex-latex] (y) -- (z);
	\end{tikzpicture}
	}
\caption{The butterfly network as a guessing game.}
\label{fig:butterfly}
\end{figure}

\subsection{Parameters of undirected graphs} \label{sec:preliminaries_graphs}

An independent set in a (simple, undirected) graph is a set of vertices where any two vertices are non-adjacent. The {\em independence number} $\alpha(G)$ of an undirected graph $G$ is the maximum cardinality of an independent set. The {\em chromatic number} $\chi(G)$ of $G$ is the smallest number of parts of a partition of its vertex set into independent sets \cite{GR01}. A graph $G$ is {\em vertex-transitive} if for all $u,v \in V$, there is an automorphism $\phi$ of $G$ such that $\phi(u) = v$. For a connected vertex-transitive graph which is neither an odd cycle nor a complete graph, we have \cite[Corollary 7.5.2]{GR01}, \cite{Lov75}.
\begin{equation} \label{eq:parameters_graph}
	\frac{|V(G)|}{\alpha(G)} \le \chi(G) \le (1 + \log \alpha(G)) \frac{|V(G)|}{\alpha(G)}.
\end{equation}

We now review three types of products of graphs; all products of two graphs $G_1$ and $G_2$ have $V(G_1) \times V(G_2)$ as vertex set. We denote two adjacent vertices $u$ and $v$ in a graph as $u \sim v$.
\begin{enumerate}
	\item In the {\em co-normal product} $G_1 \oplus G_2$, we have $(u_1,u_2) \sim (v_1,v_2)$ if and only if $u_1 \sim v_1$ or $u_2 \sim v_2$. We have $\alpha(G_1 \oplus G_2) = \alpha(G_1) \alpha(G_2)$.
		
	\item In the {\em lexicographic product} (also called composition) $G_1 \cdot G_2$, we have $(u_1,u_2) \sim (v_1,v_2)$ if and only if either $u_1 = v_1$ and $u_2 \sim v_2$, or $u_1 \sim v_1$. Although this product is not commutative, we have $\alpha(G_1 \cdot G_2) = \alpha(G_1) \alpha(G_2)$.
		
	\item In the {\em cartesian product} $G_1 \Box G_2$, we have $(u_1,u_2) \sim (v_1,v_2)$ if and only if either $u_1 = v_1$ and $u_2 \sim v_2$, or $u_2 = v_2$ and $u_1 \sim v_1$. We have $\chi(G_1 \Box G_2) = \max\{\chi(G_1), \chi(G_2)\}$ and $\alpha(G_1 \Box G_2) \le \min\{ \alpha(G_1) |V(G_2)|,\alpha(G_2) |V(G_1)|\}$.
\end{enumerate}

\section{Closure solvability and network coding} \label{sec:clD}

\subsection{Closure operators related to digraphs}

Let $D$ be a digraph on $V$.

\begin{definition} \label{def:clD}
The $D$-{\em closure} of a set of vertices $X$ is defined as follows. We let $c_D(X) = X \cup \{v \in V: v^- \subseteq X\}$ and the $D$-closure of $X$ is $\cl_D(X) := c_D^n(X)$.
\end{definition}

This definition can be intuitively explained as follows. Suppose we assign a function to each vertex of $D$, which only depends on its in-neighbourhood (the function which decides which message the vertex will transmit). If we know the messages sent by the vertices of $X$, we also know the messages which will be sent by any vertex in $c_D(X)$. By applying this iteratively, we can determine all messages sent by the vertices in $\cl_D(X)$. Therefore, $\cl_D(X)$ represents everything that is determined by $X$.

We give an alternate, easier to manipulate, definition of the $D$-closure below.

\begin{lemma} \label{lem:clD_alternate}
For any $X \subseteq V$, $Y = \cl_D(X) \backslash X$ is the largest set of vertices inducing an acyclic subgraph such that $Y^- \subseteq Y \cup X$.
\end{lemma}

\begin{IEEEproof}
First, it is clear that $Y$ is a set of vertices inducing an acyclic subgraph such that $Y^- \subseteq Y \cup X$. Conversely, suppose $Z$ induces an acyclic subgraph and $Z^- \subseteq Z \cup X$. Denoting $Z_0 = \emptyset$ and $Z_i = \{v \in Z: v^- \subseteq X \cup Z_{i-1}\}$ for $1 \le i \le n$, we have $Z_i \subseteq c_D^i(X) \backslash X$ and hence $Z = Z_n \subseteq Y$.
\end{IEEEproof}

\begin{example} \label{ex:clD}
Some special classes of digraphs yield famous closure operators (all claims follow from Lemma \ref{lem:clD_alternate}).
\begin{enumerate}
    \item If $D$ is an acyclic digraph, then $\cl_D = U_{0,n}$, i.e. $\cl_D(X) = V$ for all $X$. This can be intuitively explained by the fact that an acyclic digraph comes from a network coding instance without any source or destination: no information can then be transmitted.

    \item If $D$ is the directed cycle $C_n$, then $\cl_{C_n} = U_{1,n}$, i.e. $\cl_D(\emptyset) = \emptyset$ and $\cl_D(v) = V$ for all $v \in V$. Therefore, the solutions are $(n,1,n)$ MDS codes, such as the repetition code. Intuitively, $C_n$ comes from a network coding instance with one source and one destination, and a chain of $n-1$ intermediate nodes each transmitting a message to the next until we reach the destination.

    \item If $D$ is the clique $K_n$, then $\cl_{K_n} = U_{n-1,n}$, i.e $\cl_D(X) = X$ if $|X| \le n-2$ and $v \in \cl_D(V \backslash v)$ for all $v \in V$. Therefore, the solutions of $\cl_{K_n}$ are exactly $(n,n-1,2)$ MDS codes, such as the parity-check code. Intuitively, $K_n$ comes from a generalisation of the butterfly network, with one intermediate node receiving from all sources and transmitting to all destinations.

    \item If $D$ has a loop on each vertex, then $\cl_D = U_{n,n}$, i.e. $\cl_D(X) = X$ for all $X \subseteq V$. This comes from a network with a link from every source to its corresponding destination. 
\end{enumerate}
\end{example}

Since $\cl_D(X) = V$ if and only if $X$ is a feedback vertex set of $D$, we obtain that $\cl_D$ has rank $r_D = n - \mias(D)$. 

\subsection{Closure solvability} \label{sec:solvability}


We now define the closure solvability problem. The instance consists of a closure operator $\cl$ on $V$ with rank $r$, and of a finite alphabet $A$ with $|A| \ge 2$.

\begin{definition}
A {\em coding function} for $(\cl,A)$ is a family $f$ of $n$ partitions of $A^r$ into at most $|A|$ parts such that $f_X = f_{\cl(X)}$ for all $X \subseteq V$.
\end{definition}

The problem is to determine whether there exists a coding function for $(\cl,A)$ such that $f_V$ has $A^r$ parts. That is, we want to determine whether there exists an $n$-tuple $f = (f_1,\ldots,f_n)$ of partitions of $A^r$ in at most $|A|$ parts such that
\begin{align*}
	f_X &= f_{\cl(X)} \quad \mbox{for all } X \subseteq V,\\
	f_V &= E_{A^r}.
\end{align*}

For any partition $g$ of $A^r$, we define its entropy as
$$
	H(g) := r - |A|^{-r} \sum_i |P_i(g)| \log_{|A|} |P_i(g)|.
$$
The equality partition on $A^r$ is the only partition with full entropy $r$. Denoting $H_f(X) := H(f_X)$, we can recast the conditions above as
\begin{align*}
	H_f(v) &\le 1 \quad \mbox{for all } v \in V,\\
	H_f(X) &= H_f(\cl(X)) \quad \mbox{for all } X \subseteq V,\\
	H_f(V) &= r.
\end{align*}
Therefore, $\cl$ is solvable if and only if $H_f(V) = r$ for some coding function $f$ of $\cl$ over $A$.

The first important case is solvability of uniform matroids, which is equivalent to the existence of MDS codes.

\begin{proposition} \label{prop:U}
For all $r$, $n$, and $A$, $U_{r,n}$ is solvable over $A$ if and only if there exists an $(n,r,n-r+1)$-MDS code over an alphabet of cardinality $|A|$.
\end{proposition}

The proof follows the classical argument that a code of length $n$ with cardinality $|A|^r$and minimum distance $n-r+1$ is separable (hence the term MDS code). We shall formally prove a much more general result in Section \ref{sec:solvability_graph}, therefore we omit the proof of Proposition \ref{prop:U}.

In particular, a solution for $U_{2,n}$ is then equivalent to $n-2$ mutually orthogonal latin squares; they exist for all sufficient large alphabets. This illustrates the complexity of this problem: representing $U_{2,4}$ (i.e., determining the possible orders for two mutually orthogonal latin squares) was wrongly conjectured by Euler and solved in 1960 \cite{BSP60}. 

Combinatorial representations \cite{CGR13} were recently introduced in order to capture some of the dependency relations amongst functions. A solution for the uniform matroid corresponds to a combinatorial representation of its family of bases; however, in general this is not true. Indeed, any family of bases has a combinatorial representation, while we shall exhibit closure operators which are not solvable.

\subsection{Closure solvability and network coding solvability}

We consider a multiple unicast instance: an acyclic network $N$ with $r$ sources $s_1, \ldots, s_r$, $r$ destinations $d_1,\ldots,d_r$, and $m$ intermediate nodes, where each destination $d_i$ requests the message $x_i$ sent by $s_i$. We assume that the messages $x_i$, along with everything carried on one link, is an element of an alphabet $A$. Also, any vertex transmits the same message on all its outgoing links, i.e. we are using the circuit representation reviewed in Section \ref{sec:preliminaries}. We denote the cumulative coding functions at the nodes as $f = (f_1,\ldots,f_n)$, where the first $r$ indeces correspond to the destinations and the other $m$ indeces to the intermediate nodes, and $n=r+m$.

We now convert the network coding solvability problem into a closure solvability problem. Recall the digraph $D$ on $n$ vertices corresponding to the guessing game, reviewed in Section \ref{sec:preliminaries}.

%
%
%

Intuitively, if the destination $d_i$ is able to recover $x_i$ from the messages it receives, it is also able to recover any function $\sigma(x_i)$ of that message. Conversely, if it can recover $\pi(x_i)$ for some permutation $\pi$ of $A$, then it can recover $x_i = \pi^{-1}(\pi(x_i))$ as well. We can then relax the condition and let $d_i$ request any such $\pi(x_i)$. Viewing $x_i$ as a function from $A^r$ to $A$, sending $(x_1,\ldots,x_r)$ to $x_i$, we remark that $\pi(x_i)$ has the same kernel as $x_i$ for any permutation $\pi$. Therefore, the correct relaxation is for $d_i$ to request that the partition assigned to it be the same as that of the source $s_i$.

The relaxation above is one argument to consider partitions instead of functions. The second main argument is that the dependency relations are completely (and elegantly) expressed in terms of partitions, as illustrated in the proof of Theorem \ref{th:clD_NC}.

\begin{theorem} \label{th:clD_NC}
The network $N$ is solvable over $A$ if and only if $\cl_D$ has rank $r$ and is solvable over $A$.
\end{theorem}

\begin{proof}
Let $\bar{f}$ be a solution for $N$. Then it is easy to check that $f := (\ker \bar{f}_1, \ldots, \ker \bar{f}_n)$ is a family of partitions of $A^r$ into at most $|A|$ parts such that $f_V = E_{A^r}$ and $f_{v \cup v^-} = f_{v^-}$. As such, $f_X = f_{c_D(X)}$ for all $X \subseteq V$ and hence $f_X = f_{\cl_D(X)}$. 

Conversely, let $f$ be a solution for $\cl_D$ over $A$ and let $\bar{f}_v$ on $N$ be any collection of functions with kernels $\ker \bar{f}_{s_i} = \ker \bar{f}_{d_i} = f_i$ for all $1 \le i \le r$ and $\ker \bar{f}_v = f_v$ for all $r+1 \le v \le n$. Since $f_{v \cup v^-} = f_{v^-}$, we have that $\bar{f}_v$ only depends on $f_{v^-}$; the number of parts of $f_v$ indicates that $\bar{f}_v: A^n \to A$; finally, $f_{s_1,\ldots,s_r} = f_V =  E_{A^r}$ indicates that $\bar{f}$ is a solution for $N$.
\end{proof}

We remark that the closure operator approach differs from Riis's guessing game approach. Although it also gets rid of the source/intermediate node/receiver hierarchy and works on the same digraph, the distinction is in the fact that now $f$ corresponds to the cumulated coding functions.

\section{Main results} \label{sec:solvable}

\subsection{How to use closure solvability}

So far, we have considered any possible closure operator. Let us reduce the scope of our study by generalising some concepts arising from matroid theory. 

First, we say that a vertex is a {\em loop} if it belongs to the closure of the empty set. It is clear that removing $\cl(\emptyset)$ from $V$ does not affect solvability (any vertex from $\cl(\emptyset)$ is useless). We therefore assume that $\cl(\emptyset) = \emptyset$. In particular, we only consider digraphs with positive minimum in-degree or in other words, that have a cycle. 

Second, we say that $\cl$ is {\em separable} if for all $a,b \in V$ such that $a \notin \cl(b)$ and $b \notin \cl(a)$, we have $\cl(a) \cap \cl(b) = \emptyset$. Any matroid is separable; likewise it is easily seen that any $D$-closure is separable too. 
If $\cl$ is separable, we can further simplify the problem in a more general fashion than the so-called parallel elements in a matroid. There exists $V'$ such that $V$ is partitioned into parts $\{\cl(v'): v' \in V'\}$; for any $u \in V$, there exists $v' \in V'$ such that $\cl(u) \subseteq \cl(v')$. Again, considering the closure operator $\cl'$ on $V'$ defined by $\cl'(X') = \cl(X') \cap V'$ does not affect solvability (since if $\cl$ is solvable, then there is a solution for $\cl$ where $f_u = f_{v'}$ for all $u \in \cl(v')$). Therefore, we can always restrict ourselves to $D$-closures where $\cl_D(v) = v$ for all $v$. In other words, we have just removed all vertices of in-degree one and by-passed them instead. Clearly, these vertices of degree one are useless for network coding, as they do not bring any more combinations. The only thing they can do is forward the symbol they receive. As such, we might as well by-pass them.


There is a natural partial order on the family of closure operators of $V$. We denote $\cl_1 \le \cl_2$ if for all $X$, $\cl_1(X) \subseteq \cl_2(X)$. This partial order has maximum element $U_{0,n}$ (with $\cl(X) = V$ for all $X \subseteq V$) and minimum element $U_{n,n}$ (where $\cl(X) = X$ for all $X$). 

Any tuple $f$ of partitions of $A^r$ into at most $|A|$ parts naturally yields a closure operator on $V$: we define
\begin{align*}
    \cl_f(X) &:= \{v \in V: f_{X \cup v}  = f_X\}\\
    &= \{v \in V: H_f(X \cup v) = H_f(X)\}.
\end{align*}

\begin{proposition} \label{prop:cl1<cl2}
$f$ is a coding function for $\cl$ if and only if $\cl \le \cl_f$. Therefore, if $\cl_1 \le \cl_2$ have the same rank and $\cl_2$ is solvable over $A$, then $\cl_1$ is solvable over $A$.
\end{proposition}

\begin{IEEEproof}
If $f$ is a coding function for $\cl$, then $f_{\cl(X)} = f_{X \cup v} = f_X$ for all $v \in \cl(X)$ and hence $\cl \le \cl_f$. Conversely, if $\cl(X) \subseteq \cl_f(X)$, then denote $\cl(X) = \{v_1,\ldots,v_k\}$ and $f_{\cl(X)} = f_{X \cup v_1} \vee f_{v_2,\ldots,v_k} = f_X \vee f_{v_2,\ldots, v_k} = \ldots = f_X$.

Since $\cl_2$ is solvable, there exists a coding function $f$ for $\cl_2$ with entropy $r$, where $r$ is the rank of $\cl_1$ and $\cl_2$. But then $\cl_1 \le \cl_2 \le \cl_f$ and hence $f$ is also a solution for $\cl_1$.
\end{IEEEproof}

If $\cl$ is a matroid, the solvability problem is equivalent to determining whether they form a secret-sharing matroid, i.e. whether there exists a scheme whose access structure is the family of spanning sets of that matroid.

\begin{theorem} \label{th:solvable_matroid}
If $\cl$ is a matroid, then $\cl$ is solvable over some alphabet if and only if it is a secret-sharing matroid.
\end{theorem}

\begin{IEEEproof}
By definition, a secret-sharing matroid is solvable over some alphabet. Conversely, let $f$ be a solution for $\cl$. Let $\rk$ be the rank function associated to $\cl$, i.e. $\rk(X) = \min\{|b|: \cl(b) = \cl(X)\}$ and $\cl(X) = \{v \in V: \rk(X) = \rk(X \cup v)\}$ \cite{Oxl06}. Then for any $X$, we have $H_f(X) = H_f(b) \le |b| = \rk(X)$. Moreover, there exists $Y$ such that $\cl(X \cup Y) = V$ and $\rk(Y) + \rk(X) = r$, hence $H_f(X) \ge H_f(V) - H_f(Y) \ge r - \rk(Y) \ge \rk(X)$. Thus, $H_f(X) = \rk(X)$ for all $X$ and $\cl_f(X) = \{v \in V: \rk(X) = \rk(X \cup v)\} = \cl(X)$.
\end{IEEEproof}

\subsection{Solvable networks} \label{sec:rank_2}

In this subsection, we apply the conversion of network coding solvability in order to closure solvability to determine that some classes of networks are solvable. Using general closure operators allows us to think outside of networks. In particular, it allows us to use uniform matroids, which have been proved to be solvable over many alphabets (see Proposition \ref{prop:U}), but which do not arise from networks in general (see Proposition \ref{prop:Urn} below).

\begin{proposition} \label{prop:Urn}
The uniform matroid $U_{r,n}$ is the $D$-closure of a digraph $D$ if and only if $r \in \{0,1,n-1,n\}$.
\end{proposition}

\begin{IEEEproof}
The cases $r =0,1,n-1,n$ respectively have been illustrated in Example \ref{ex:clD}. Conversely, suppose a digraph has $D$-closure $U_{r,n}$, where $2 \le r \le n-2$. Then any set of $n-r$ vertices induces an acyclic subgraph, while any set of $n-r+1$ vertices induces a cycle. This implies that any set of $n-r$ vertices induces a (directed) path. Without loss, let $v_1,\ldots,v_{n-r}$ induce a path (in that order), then $v_1,\ldots,v_{n-r},v_{n-r+1}$ induce a cycle, and so do $v_1,\ldots,v_{n-r},v_{n-r+2}$. Therefore, in the subgraph induced by $v_2,\ldots,v_{n-r+2}$, the vertex $v_{n-r}$ has out-degree $2$ and hence that graph is not a cycle.
\end{IEEEproof}

We can then prove that all digraphs with minimum degree equal to the rank, or with rank $2$, are solvable. Note that the case of rank $2$ has already been proved in \cite{WS10} using a much longer argument.

\begin{theorem} \label{th:rank_2}
Any $\cl_D$ of rank $2$ is solvable over all sufficiently large alphabets. Moreover, if the mininum in-degree of $D$ is equal to its rank, then $\cl_D$ is solvable by linear functions over all sufficiently large prime powers.
\end{theorem}

\begin{IEEEproof}
The simplifications above mean that we can assume $\cl_D(v) = v$ for all $v \in V$. This is equivalent to $\cl_D \le U_{2,n}$, therefore by Proposition \ref{prop:cl1<cl2}, any digraph with rank 2 is solvable whenever $U_{2,n}$ is.

Moreover, suppose the minimum in-degree is equal to the rank $r$. Then for any $X \subseteq V$ with $|X| \le r$, we have $c_D(X) = X$ and hence $\cl_D(X) = X$; therefore, $\cl_D \le U_{r,n}$. Again Proposition \ref{prop:cl1<cl2} yields the result.
\end{IEEEproof}

\section{Combining closure operators} \label{sec:combining}

In this section, we let $V_1, V_2 \subseteq V$ with respective cardinalities $n_1$ and $n_2$ such that $V_1 \cap V_2 = \emptyset$ and $V_1 \cup V_2 = V$. For any $X \subseteq V$, we denote $X_1 = X \cap V_1$ and $X_2 = X \cap V_2$. We also let $\cl_1, \cl_2$ be closure operators of rank $r_1$ and $r_2$ over $V_1$ and $V_2$, respectively.

\subsection{Disjoint and unidirectional unions}

We first generalise some definitions from matroid theory \cite{Oxl06}.

\begin{definition} \label{def:clS}
For any closure operator $\cl$ and any $V_2 \subseteq V$, the {\em deletion of $V_2$} and the {\em contraction of $V_2$} from $\cl$ are the closure operators defined on $V_1$ by
\begin{align*}
	\cl \backslash^{V_2} (X) &:= \cl(X) \backslash V_2\\
	\cl /_{V_2} (X) &:= \cl(X \cup V_2) \backslash V_2
\end{align*}
for any $X \subseteq V_1$.
\end{definition}



\begin{proposition} \label{prop:clD|S}
If $\cl_D$ is the closure operator associated to the digraph $D$, then for any $V_2 \subseteq D$, $\cl_{D[V_1]} = \cl_D /_{V_2}$, where $D[V_1]$ is the digraph induced by the vertices in $V_1$. Thus $r(\cl_D /_{V_2}) = |V_1| - \mias(D[V_1])$ for any $V_1$.
\end{proposition}

\begin{IEEEproof}
Let $X \subseteq V_1$, then any subset $Y$ of $V_1 \backslash X = V \backslash (X \cup V_2)$ induces an acyclic subgraph of $D$ if and only if it induces an acyclic subgraph of $D[V_1]$; moreover, $Y^- \subseteq X \cup Y$ in $D[V_1]$ if and only if $Y^- \subseteq X \cup Y \cup V_2$. By Lemma \ref{lem:clD_alternate}, we obtain $\cl_{D[V_1]}(X) \backslash X = \cl_D(X \cup V_2) \backslash (X \cup V_2)$ and hence $\cl_D / V_2(X) = \cl_{D[V_1]}(X)$.
\end{IEEEproof}


\begin{definition} \label{def:unions}
The {\em disjoint union} and {\em unidirectional union} of $\cl_1$ and $\cl_2$ are closure operators on $V$ respectively given by
\begin{align*}
	\cl_1 \cup \cl_2(X) &:= \cl_1(X_1) \cup \cl_2(X_2)\\
	\cl_1 \vcup \cl_2(X) &:= \begin{cases}
	V_1 \cup \cl_2(X_2) &\mbox{if } \cl_1(X_1) = V_1\\
	\cl_1(X_1) \cup X_2 &\mbox{otherwise}.
	\end{cases}
\end{align*}
\end{definition}

For any $\cl_1$, $\cl_2$ we have $\cl_1 \vcup \cl_2 \le \cl_1 \cup \cl_2$ and
$$
	r(\cl_1 \cup \cl_2) = r(\cl_1 \vcup \cl_2) = r_1 + r_2.
$$
Recall the definitions of unions of digraphs in Section \ref{sec:preliminaries}. Our definitions were tailored such that
\begin{align*}
	\cl_{D_1 \cup D_2} &= \cl_{D_1} \cup \cl_{D_2}\\
	\cl_{D_1 \vcup D_2} &= \cl_{D_1} \vcup \cl_{D_2}.
\end{align*}

Moreover, if there is a loop on vertex $v$ in the digraph $D$, then $\cl_D(X) = \cl_D(X \backslash v) \cup (X \cap v)$, or in other words, $\cl_D = \cl_{D[V \backslash v]} \cup U_{1,1} = \cl_{D[V \backslash v]} \vcup U_{1,1}$. We also remark that if $\cl_1$ and $\cl_2$ are matroids, then $\cl_1 \cup \cl_2$ is commonly referred to as the {\em direct sum} of $\cl_1$ and $\cl_2$ \cite{Oxl06}.

The disjoint and unidirectional unions are related to the contraction as follows.

\begin{proposition} \label{prop:cl_strong}
For any $\cl$ and any $V_2 \subseteq V$, the following are equivalent
\begin{enumerate}
	\item $\cl/_{V_2} = \cl \backslash^{V_2}$, i.e. for all $X \subseteq V$, $\cl(X) \cap V_1 = \cl(X \cup V_2) \cap V_1$;
	
	\item $\cl/_{V_2} \vcup \cl/_{V_1} \le \cl \le \cl/_{V_2} \cup \cl/_{V_1}$;
	
	\item there exist $\cl_1, \cl_2$ defined on $V_1$ and $V_2$ respectively such that
	$$
		\cl_1 \vcup \cl_2 \le \cl \le \cl_1 \cup \cl_2.
	$$
\end{enumerate}
\end{proposition}

\begin{IEEEproof}
The first property implies the second, due to the following pair of inequalities: For any $V_1$,
$$
	\cl \backslash^{V_2} \vcup \cl /_{V_1} \le \cl \le \cl /_{V_2} \cup \cl/_{V_1}.
$$
To prove the first inequality, we have
$$
	\cl \backslash^{V_2} \vcup \cl /_{V_1} (X) = \begin{cases}
	V_1 \cup \cl/_{V_1}(X_2) = \cl(X_2 \cup V_1) \subseteq \cl(X) &\mbox{if } V_1 \subseteq \cl(X_1)\\
	(\cl(X_1) \cap V_1) \cup X_2 \subseteq \cl(X) &\mbox{otherwise}.
	\end{cases}
$$
For the second inequality, we have
$$
	\cl(X) \backslash V_2 = \cl(X_1 \cup X_2) \backslash V_2 \subseteq \cl(X_1 \cup V_2) \backslash V_2 = \cl/_{V_2}(X_1),
$$
and similarly $\cl(X) \backslash V_1 \subseteq \cl/_{V_1}(X_2)$, and hence $\cl(X) \subseteq \cl/_{V_2} \cup \cl/_{V_1}(X)$.

Clearly, the second property implies the third one. Finally, if there exist such $\cl_1$ and $\cl_2$, then it is easy to check that $\cl_1 = \cl\backslash^{V_2} = \cl/_{V_2}$.
\end{IEEEproof}

\subsection{Application to removing useless vertices}

The first property of Proposition \ref{prop:cl_strong} indicates that $V_2$ has no effect on $V_1$; thus suggesting the following notation.

\begin{definition} \label{def:weak}
If there exists $V_2$ such that $\cl/_{V_2} = \cl \backslash^{V_2}$, we say that $\cl$ is disconnected and that $V_2$ is {\em weak}. If $V_2$ is weak and acyclic, then we say $V_2$ is {\em useless}.
\end{definition}

The $D$-closure of a non strongly connected graph is disconnected. However, there are strongly connected graphs whose $D$-closure is disconnected, for instance if there is a loop on a vertex, or in the graph in Figure \ref{fig:D}, where $\cl \backslash^{V_2} = \cl/_{V_2}$ for $T = 45$. 

\begin{example}
The canonical example of a strongly connected graph with disconnected closure operator is given in Figure \ref{fig:123_disconnected}. On that graph, $\{3\}$ is useless, for 
\begin{align*}
	\cl_D(\emptyset) \backslash 3 &= \cl_D(3) \backslash 3 = \emptyset\\
	\cl_D(1) \backslash 3 &= \cl_D(13) \backslash 3 = 12\\
	\cl_D(2) \backslash 3 &= \cl_D(23) \backslash 3 = 12\\
	\cl_D(12) \backslash 3 &= \cl_D(123) \backslash 3 = 12
\end{align*}

\begin{figure}
\begin{center}
\begin{tikzpicture}
    \tikzstyle{every node}=[draw,shape=circle];

	\node (1) at (0,2) {1};
	\node (2) at (0,0) {2};
	\node (3) at (1,1) {3};
	
    \draw[latex-latex] (1) -- (2);
    \draw[latex-latex] (2) -- (1);
    \draw[-latex] (1) -- (3);
    \draw[-latex] (3) -- (2);
\end{tikzpicture}
\end{center}
\caption{A graph which is strongly connected but whose $D$-closure is disconnected.} \label{fig:123_disconnected}
\end{figure}
\end{example}

\begin{figure}
\begin{center}
\begin{tikzpicture}
    \tikzstyle{every node}=[draw,shape=circle];

	\node (1) at (0,2) {1};
	\node (2) at (0,0) {2};
	\node (3) at (1,1) {3};
	\node (4) at (3,1) {4};
	\node (5) at (2,1) {5};
	
    \draw[latex-latex] (1) -- (2);
    \draw[latex-latex] (2) -- (3);
    \draw[latex-latex] (3) -- (1);
    \draw[-latex] (1) .. controls(2,0) .. (4);
    \draw[-latex] (2) .. controls(2,2) .. (4);
    \draw[-latex] (4) -- (5);
    \draw[-latex] (5) -- (3);

\end{tikzpicture}
\end{center}
\caption{A graph which is strongly connected but whose $D$-closure is disconnected.} \label{fig:D}
\end{figure}

\begin{proposition} \label{prop:connected_not_strong}
Suppose $D$ is strongly connected, then $V_2$ is weak for $\cl_D$ if and only if it is useless.
\end{proposition}

\begin{proof}
We first claim that all arcs from $V_2$ to $V_1$ come from $\cl_{D[V_2]}(\emptyset)$. Indeed, let $u \in V_1$ such that $u^- \cap V_1 \ne \emptyset$. Then $u \in \cl_D/_{V_2}(V_1 \backslash u) = \cl_D\backslash^{V_2}(V_1 \backslash u)$, and hence $u \in \cl_D(V_1 \backslash u)$. Since $u^- \subseteq \cl_D(V_1 \backslash u)$, the intersection $X := \cl_D(V_1 \backslash u) \cap V_2$ is not empty. By Lemma \ref{lem:clD_alternate}, $X$ induces an acyclic subgraph and $X^- \subseteq V_1 \cap X$, which is equivalent to $X \subseteq \cl_{D[V_2]}(\emptyset)$.

Now, suppose $V_2$ is not acyclic, i.e. $V_2 \ne \cl_{D[V_2]}(\emptyset)$. But then, by the claim above there are no arcs from $V_2 \backslash \cl_{D[V_2]}(\emptyset)$ to its complement, and $D$ is not strongly connected.
\end{proof}

As a corollary, if $D$ is an undirected graph, then $\cl_D$ is connected if and only if $D$ is connected.

We remark that if $V_2$ is weak, then $V_2$ is closed, for $\cl(V_2) \backslash V_2 = \cl(\emptyset) \backslash V_2 = \emptyset.$ Also, it is easy to check that the union of two weak sets is weak, hence there exists a largest weak set. Thus, if $D$ is strongly connected, there exists a largest useless set, referred to as the {\em useless part} of $D$.

\begin{algorithm}

\caption{Remove the useless part of a strongly connected digraph $D$}
\label{alg:1}

\begin{algorithmic}
\STATE $T \gets T(D)$
\REPEAT
	\STATE $Found \gets 0$
	\WHILE{$v \in T$ \AND $Found = 0$}
		\STATE $Found \gets |v^-|$
		\WHILE[Check that $\{v\}$ is useless]{$u \in v^-$ \AND $Found = 1$}
			\IF{$v \notin \cl_D(u^- \backslash v)$}
				\STATE $Found = 0$
			\ENDIF
		\ENDWHILE
		\IF[Remove $v$]{$Found > 0$}
			\STATE $V \gets V \backslash v$
			\STATE $T \gets T \backslash v$
		\ENDIF
	\ENDWHILE
\UNTIL{$Found = 0$}
\RETURN $D$
\end{algorithmic}
\end{algorithm}

We say a cycle $v_1,\ldots,v_k$ is {\em chordless} if there does not exist $i, j \in \{1,\ldots,k\}$, $(i,j) \ne (1,k)$, such that $v_i,\ldots, v_j$ is a cycle. In other words, a chordless cycle does not cover another shorter cycle. Then $T(D)$ is the set of vertices which do not belong to any chordless cycle.

\begin{theorem} \label{th:algorithm}
Algorithm \ref{alg:1} removes the useless part of a strongly connected digraph in polynomial time.
\end{theorem}

The proof of Theorem \ref{th:algorithm} is given in Appendix \ref{app:th}.

\subsection{Bidirectional union}

\begin{definition}
The {\em bidirectional union} of $\cl_1$ and $\cl_2$ is defined as
$$
	\cl_1 \bcup \cl_2(X) := \begin{cases}
	V_1 \cup \cl_2(X_2) & \mbox{if } X_1 = V_1\\
	\cl_1(X_1) \cup V_2 & \mbox{if } X_2 = V_2\\
	X_1 \cup X_2 & \mbox{otherwise}.
	\end{cases}
$$
\end{definition}

It is easily shown that $\cl_1 \bcup \cl_2 \le \cl_1 \vcup \cl_2$ and $r(\cl_1 \bcup \cl_2) = \min\{r_1 + n_2, r_2 + n_1\}$. Moreover, for any $\cl$ and any $V_1 \subseteq V$, we have
$$
	\cl/_{V_2} \bcup \cl/_{V_1} \le \cl \le \cl/_{V_2} \cup \cl/_{V_1}.	
$$
The first inequality means that the bidirectional union is the way to combine $\cl_1$ and $\cl_2$ which brings the fewest dependencies; as such it is the union of $\cl_1$ and $\cl_2$ with the highest entropy.

The bidirectional union of digraphs does correspond to the bidirectional union of closure operators: 	
$$
	\cl_{D_1 \bcup D_2} = \cl_{D_1} \bcup \cl_{D_2},
$$
and the converse is given below.

\begin{proposition} \label{lem:bar_cup_graphs}
If $\cl = \cl_1 \bcup \cl_2$, then $\cl_1 = \cl/_{V_1}$ and $\cl_2 = \cl/_{V_2}$. Moreover, if $D$ is a loopless graph, then $\cl_D = \cl_1 \bcup \cl_2$ if and only if $\cl_1 = \cl_{D[V_1]}$, $\cl_2 = \cl_{D[V_2]}$, and $D = D[V_1] \bcup D[V_2]$.
\end{proposition}

\begin{IEEEproof}
The first claim is easy to prove. For the second claim, if $\cl_D = \cl_1 \bcup \cl_2$, then $\cl_1 = \cl_{D[V_1]}$ and $\cl_2 = \cl_{D[V_2]}$. Suppose the arc $(v_1,v_2)$ is missing between $V_1$ and $V_2$. Then $\cl_D(V \backslash \{v_1,v_2\}) = V$ (since $\{v_1,v_2\}$ is acyclic), while $\cl_D/_{V_2} \bcup \cl_D/_{V_1}(V \backslash \{v_1,v_2\}) = V \backslash \{v_1,v_2\}$. The converse is trivial.
\end{IEEEproof}

\section{Guessing number and solvability graph} \label{sec:solvability_graph}

\subsection{Definition and main results}

The solvability graph extends the definition of the so-called guessing graph to all closure operators. Most of this section naturally extends \cite{GR11}. Therefore, we shall omit certain proofs which are very similar to their counterparts in \cite{GR11}.

First of all, we need the counterpart of the guessing number of a graph for closure operators. Any partition $f_i$ of $A^r$ into at most $|A|$ parts is henceforth denoted as $f_i = \{P_a(f_i): a \in A\}$, where some parts $P_a(f_i)$ are possibly empty. By extension, for any tuple $f = (f_1,\ldots,f_n)$, the partition $f_V$ is denoted as $f_V = \{P_x(f_V) : x \in A^n \}$, where $P_x(f_V) = \bigcap_{i=1}^n P_{x_i}(f_i)$. We denote the set of words of $A^n$ indexing non-empty parts of $f_V$ as the image of $f$:
$$
	\mbox{Im}(f) := \{x \in A^n: P_x(f_V) \ne \emptyset\}.
$$

\begin{definition} \label{def:g}
The guessing number of $\cl$ over $A$ is given by
$$
	g(\cl,A) := \max \{\log_{|A|} |\mbox{Im}(f)| : f \mbox{ coding function for } \cl \mbox{ over } A \}.
$$
\end{definition}

A coding function has an image of size $|A|^r$ if and only if it is a solution; therefore, $\cl$ is solvable over $A$ if and only if $g(\cl,A) = r$.

Next, we define the solvability graph of closure operators.

\begin{definition} \label{def:G}
The solvability graph $\mathrm{G}(\cl,A)$ has vertex set $A^n$ and two words $x, y \in A^n$ are adjacent if and only if there exists no coding function $f$ for $\cl$ over $A$ such that $x,y \in \mbox{Im}(f)$.
\end{definition}

Proposition \ref{prop:properties_guessing_graph} below enumerates some properties of the solvability graph. In particular, Property \ref{it:E} provides a concrete and elementary description of the edge set which makes adjacency between two configurations easily decidable.

\begin{proposition} \label{prop:properties_guessing_graph}
The solvability graph $\mathrm{G}(\cl,A)$ satisfies the following properties:
\begin{enumerate}
	\item \label{it:N(D,s)} It has $|A|^n$ vertices.

    \item \label{it:E} Its edge set is $E = \bigcup_{S \subseteq V, v \in \cl(S)} E_{v,S}$, where $E_{v,S} = \{xy : x_S = y_S, x_v \ne y_v\}$.
	
    \item \label{it:vt} It is vertex-transitive.
\end{enumerate}
\end{proposition}

\begin{IEEEproof}
Property \ref{it:N(D,s)} follows from the definition. We now prove Property \ref{it:E}. $f$ is a coding function if and only if $f_{S \cup v} = f_S$ for all $S \subseteq V$ and any $v \in \cl(S)$, which in turn is equivalent to $P_{x_{S \cup v}} (f_{S \cup v}) = P_{x_S}(f_S)$ for all $x \in \mbox{Im}(f)$. Therefore, if $x_S = y_S$, $v \in \cl(S)$ and $f$ is a coding function, we have $P_{x_S}(f_{x_S}) \subseteq P_{x_v}(f_{x_v})$ and $P_{x_S}(f_{x_S}) \subseteq P_{y_v}(f_{y_v})$, or in other words $x_v = y_v$.

Conversely, for distinct $xy \notin E$ we construct the following coding function. Let $g$ be a partition of $A^r$ into two nonempty parts and for all $v \in V$, let $f_v$ have two parts $P_{x_v} = P_1(g)$ and $P_{y_v} = P_2(g)$ if $x_v \ne y_v$ and $f_v$ have one part otherwise. Then, $\mbox{Im}(f) = \{x,y\}$; for all $T \subseteq V$, $P_{x_T}(f_T) = P_{y_T}(f_T)$ if and only if $x_T = y_T$; and $f_V = g$. We now check that $f$ is indeed a coding function: let $S \subseteq V$ and $v \in \cl(S)$. If $x_S = y_S$, then $x_v = y_v$ and hence $f_S = f_{S \cup v}$ has one part. Otherwise $f_S = g = f_V = f_{S \cup v}$.

For Property \ref{it:vt}, we remark that $\mathrm{G}(\cl,A)$ is a Cayley graph \cite{GR01}, hence it is vertex-transitive. More explicitly, if we let $A = \mathbb{Z}_{|A|}$, then $\phi(z) = z - x + y$ is an automorphism of the solvability graph which takes $x$ to $y$ for any $x,y \in A^n$.
\end{IEEEproof}

\begin{corollary} \label{cor:U}
The solvability graph for the uniform matroid $U_{r,n}$ has edge set $E = \{xy: d_H(x,y) \le n-r\}$.
\end{corollary}

The main reason to study the solvability graph is given in Theorem \ref{th:G} below.

\begin{theorem} \label{th:G}
A set of words in $A^n$ is an independent set of $\mathrm{G}(\cl,A)$ if and only if they are the image of $A^r$ by a coding function for $\cl$ over $A$.
\end{theorem}

\begin{IEEEproof}
By definition of the solvability graph, the image of a coding function forms an independent set. Conversely, let $\{x^i\}_{i=1}^k$ be an independent set of the solvability graph $\mathrm{G}(\cl,A)$. Let $g$ be a partition of $A^r$ into $k$ nonempty parts and let
$$
	P_{x_v^i}(f_v) := \bigcup_{j: x_v^j = x_v^i} P_j(g).
$$
Then we have $\mbox{Im}(f) = \{x^i\}_{i=1}^k$; for all $T \subseteq V$, $P_{x^i_T} = P_{x^j_T}$ if and only if $x^i_T = x^j_T$; and $f_V = g$. We now justify that $f$ is a coding function. Let $v \in \cl(S)$, then for all $i$, $S_i := \{j: x_S^j = x_S^i\} = \{j: x_{S \cup v}^j = x_{S \cup v}^i \}$ and hence
$$
	P_{x_{S \cup v}^i} = \bigcup_{j \in S_i} P_j(g) = P_{x_S^i},
$$
which means $f_S = f_{S \cup v}$.
\end{IEEEproof}

\begin{corollary} \label{cor:G}
We have $\log_{|A|} \alpha(\mathrm{G}(\cl,A)) = g(\cl,A)$ and hence $\alpha(\mathrm{G}(\cl,A)) = |A|^r$ if and only if $\cl$ is solvable over $A$.
\end{corollary}

In \cite{GR11}, we remark that the index coding problem asks for the chromatic number of the guessing graph of a digraph. We can extend the index coding problem to any closure operator and we say that $\cl$ is index-solvable over $A$ if $b(\cl,A) := \log_{|A|} \chi(\mathrm{G}(\cl,A)) = n-r$. We have
\begin{align*}
	g(\cl,A) + b(\cl,A) &\ge n,\\ 
	\lim_{|A| \to \infty} g(\cl,A) + \lim_{|A| \to \infty} b(\cl,A) &= n
\end{align*}
by (\ref{eq:parameters_graph}). Therefore, although determining $g(\cl,A)$ and $b(\cl,A)$ are distinct over a fixed alphabet $A$, they are asymptotically equivalent. More strikingly, solvability and index-solvability are equivalent for finite alphabets too, as seen below.

\begin{theorem} \label{th:index-solvable}
The closure operator $\cl$ is solvable over $A$ if and only if it is index-solvable over $A$.
\end{theorem}

\begin{IEEEproof}
Let $\{x^i\}$ be an independent set of $\mathrm{G}$ and $b$ be a basis of $\cl$. Without loss, let $b = \{1,\ldots,r\}$). First, we remark that $x^i_b \ne x^j_b$ for all $i \ne j$, for otherwise $x^i_{V \backslash b} \ne x^j_{V \backslash b}$ and $x^i_b = x^j_b$ means that $x^i \sim x^j$. Secondly, let $A = \mathbb{Z}_{|A|}$, then for any $w \in A^{n-r}$ and any $i$, denote $x^i + w = (x^i_b,x^i_{V \backslash b} + w)$. Then it is easily shown that $S_w = \{x^i + w\}$ forms an independent set and that the family $\{S_w\}$ forms a partition of $A^n$ into $|A|^{n-r}$ independent sets.

Conversely, if $\chi(G(\cl,A)) = |A|^{n-r}$, then $\alpha(G(\cl,A)) = |A|^r$ by (\ref{eq:parameters_graph}).
\end{IEEEproof}

\subsection{Neighbourhood and girth}

Note that the relation ``having an arc from $u$ to $v$'' cannot be expressed in terms of the $D$-closure. Indeed, all acyclic graphs on $n$ vertices, from the empty graph to an acyclic tournament, all have the same closure operator $U_{0,n}$. However, the $D$-closure of the in-neighbourhood of a vertex can be described by means of the digraph closure.

\begin{lemma}
For any $v$ and any $X \subseteq V \backslash \{v\}$,  $v \in \cl_D(X)$ if and only if $\cl_D(v^-) \subseteq \cl_D(X)$.
\end{lemma}

\begin{IEEEproof}
Suppose $v \in \cl_D(X) \backslash X$, then $Y = \cl_D(X) \backslash X$ induces an acyclic subgraph and $Y^- \subseteq \cl_D(X)$; in particular, $v^- \subseteq \cl_D(X)$. Since $v \in \cl_D(v^-)$, we easily obtain the converse.
\end{IEEEproof}

We remark that if there is a loop on $v$, then there exists no set $X \subseteq V \backslash \{v\}$ such that $v \in \cl_D(X)$. Note that $v^-$ is not necessarily an inner basis of its own closure, for instance this is trivial in nonempty acyclic digraphs.

Based on our results about closure operators associated to digraphs, we can define some concepts to any closure operators which generalise those of digraphs.

\begin{definition}
For any vertex $v$, the {\em degree} of $v$ is
$$
	d_v := \min\{|X| :v \in \cl(X) \backslash X\}
$$
if there exists such set $X$, or by convention is equal to $0$ otherwise. We denote the minimum degree as $\delta$. 
\end{definition}

Note that the degree (according to the closure operator $\cl_D$) of a vertex of the digraph $D$ is not necessarily equal to the size of its in-neighbourhood. 

\begin{definition}
We say a subset $X$ of vertices is {\em acyclic} if $\cl(V \backslash X) = V$. The {\em girth} $\gamma$ of the closure operator as the minimum size of a non-acyclic subset of vertices. 
\end{definition}

Here, the girth of a digraph is equal to the girth of its closure operator.

We denote the maximum cardinality of a code over $A$ of length $n$ and minimum distance $d$ as $M_A(n,d)$.

\begin{proposition} \label{prop:bounds}
For any $\cl$, we have
$$
	\log_{|A|} M_A(n,n-\delta+1) \le g(\cl,A) \le \log_{|A|} M_A(n,\gamma).
$$
\end{proposition}

Since $\delta \le r$ and $\gamma \le n-r+1$, we have $\gamma = n-\delta+1$ if and only if $\cl = U_{r,n}$.

\subsection{Combining closure operators}

Recall the definitions of unions of closure operators in Section \ref{sec:combining}. The following theorem is the counterpart of Propositions 6, 7, and 8 in \cite{GR11}.

\begin{theorem} \label{th:G(union)}
For any $\cl_1$ and $\cl_2$ defined on disjoint sets $V_1$ and $V_2$ of cardinalities $n_1$ and $n_2$, we have
\begin{align*}
	G(\cl_1 \cup \cl_2,A) &= G(\cl_1,A) \oplus G(\cl_2,A)\\
	G(\cl_1 \vcup \cl_2,A) &= G(\cl_1,A) \cdot G(\cl_2,A)\\
	G(\cl_1 \bcup \cl_2,A) &= G(\cl_1,A) \Box G(\cl_2,A).
\end{align*}
Therefore,
\begin{align*}
	g(\cl_1 \cup \cl_2,A) &= g(\cl_1 \vcup \cl_2,A) = g(\cl_1,A) + g(\cl_2,A)\\
	b(\cl_1 \bcup \cl_2,A) &= \max\{b(\cl_1,A), b(\cl_2,A)\}\\
	g(\cl_1 \bcup \cl_2,A) &\le \min\{g(\cl_1,A) + n_2, g(\cl_2,A) + n_1\}.
\end{align*}
\end{theorem}

\begin{corollary} \label{cor:strong}
The following are equivalent:
\begin{itemize}
	\item $\cl_1$ and $\cl_2$ are solvable over $A$
	
	\item $\cl_1 \cup \cl_2$ is solvable over $A$
	
	\item $\cl_1 \vcup \cl_2$ is solvable over $A$.
\end{itemize}
\end{corollary}

Therefore, when studying solvability, we can only consider connected closure operators.

\begin{corollary}
Without loss, suppose $n_1 - r_1 \ge n_2 - r_2$, then we have the following list of properties each implying the next:
\begin{itemize}
	\item $\cl_1$ and $\cl_2$ are solvable over $A$;
	
	\item $\cl_1$ is solvable over $A$ and $b(\cl_2,A) \le n_1 - r_1$;
	
	\item $\cl_1 \bcup \cl_2$ is solvable over $A$;
	
	\item $\cl_1$ is solvable over $A$ and $g(\cl_2,A) \ge n_2 -n_1 + r_1$.
\end{itemize}
In particular, if $n_1-r_1 = n_2-r_2$, then $\cl_1 \bcup \cl_2$ is solvable over $A$ if and only if $\cl_1$ and $\cl_2$ are solvable over $A$.
\end{corollary}

An example where $\cl_2$ is not solvable, yet $\cl_1 \bcup \cl_2$ is solvable, is given in Figure \ref{fig:E3cupC5}.

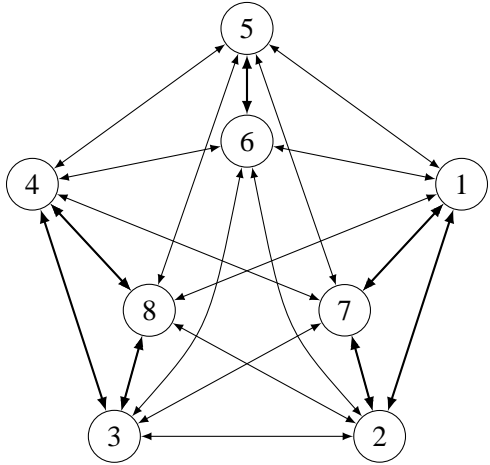
\begin{figure}
\begin{center}
\begin{tikzpicture}
    \tikzstyle{every node}=[draw,shape=circle];

\foreach \x in {1,...,5}{
	\node (\x) at (90-72*\x:3) {\x};
}

    \draw[thick, latex-latex] (1) -- (2);
    \draw[latex-latex] (2) -- (3);
    \draw[thick, latex-latex] (3) -- (4);
    \draw[latex-latex] (4) -- (5);
    \draw[latex-latex] (5) -- (1);

\foreach \x in {6,7,8}{
	\node (\x) at (90-120*\x:1.5) {\x};
}

\foreach \y in {1,...,5}{
\foreach \x in {7,8}{
    \draw[latex-latex] (\y) -- (\x);
}
}

    \draw[latex-latex] (1) -- (6);
	\draw[latex-latex] (2) .. controls(0.5,-1) .. (6);
	\draw[latex-latex] (3) .. controls(-0.5,-1) .. (6);
    \draw[latex-latex] (4) -- (6);
	\draw[thick, latex-latex] (5) -- (6);

	\draw[thick, latex-latex] (1) -- (7);
	\draw[thick, latex-latex] (2) -- (7);
	\draw[thick, latex-latex] (3) -- (8);
	\draw[thick, latex-latex] (4) -- (8);
	
\end{tikzpicture}
\end{center}
\caption{The bidirectional union $E_3 \bcup \bar{C}_5$. The vertices of $\bar{C}_5$ form a basis; the highlighted disjoint cliques $127, 248, 56$ show that it is solvable.} \label{fig:E3cupC5}
\end{figure}

The results on the bidirectional union can be viewed as ``network sharing,'' illustrated in Figure \ref{fig:sharing}. Suppose we have two solvable networks $N_1$ and $N_2$, where $N_1$ has the same number of or more intermediate nodes than $N_2$. Then $N_2$ can be plugged in to $N_1$, which can share its links with $N_2$ without compromising its solvability. In the resulting shared network, not only each source-destination pair of $N_2$ is there, but also each intermediate node yields an additional source-destination pair. As a result, the only intermediate nodes are those coming from $N_1$.

\begin{figure}
\centering
	\subfloat[First network $N_1$]{\begin{tikzpicture}
		\tikzstyle{every node}=[draw,shape=circle];

		\node (s1) at (0,4) {$s_1$};
		\node (d1) at (4,0) {$d_1$};
		
		\node (s2) at (4,4) {$s_2$};
		\node (d2) at (0,0) {$d_2$};
		
		\node (v3) at (2,2) {$v_3$};
		
		\draw[-latex] (s1) -- (d2);
		\draw[-latex] (s1) -- (v3);
		
		\draw[-latex] (s2) -- (d1);
		\draw[-latex] (s2) -- (v3);
		
		\draw[-latex] (v3) -- (d1);
		\draw[-latex] (v3) -- (d2);	
	\end{tikzpicture}}
	\hspace{2cm}
	\subfloat[Second network $N_2$]{\begin{tikzpicture}
		\node (a) at (-2,0) {};
		\node (b) at (2,0) {};
		
		\node[draw,shape=circle] (s4) at (0,4) {$s_4$};
		\node[draw,shape=circle] (d4) at (0,0) {$d_4$};
		
		\node[draw,shape=circle] (v5) at (0,2) {$v_5$};
		
		\draw[-latex] (s4) -- (v5);
		\draw[-latex] (v5) -- (d4);
	\end{tikzpicture}}
	\hspace{2cm}
	\subfloat[Shared network]{\begin{tikzpicture}
		\tikzstyle{every node}=[draw,shape=circle];

		\node (s1) at (0,7) {$s_1$};
		\node (d1) at (4,1) {$d_1$};
		
		\node (s2) at (4,7) {$s_2$};
		\node (d2) at (0,1) {$d_2$};
		
		\node (s4) at (-4,8) {$s_4$};
		\node (d4) at (-4,0) {$d_4$};

		\node (s5) at (8,8) {$s_5$};
		\node (d5) at (8,0) {$d_5$};
		
		\node (v3) at (2,4) {$v_3$};
		
		\draw[-latex] (s1) -- (d2);
		\draw[-latex] (s1) -- (v3);
		\draw[-latex] (s1) -- (d4);
		\draw[-latex] (s1) -- (d5);
		
		\draw[-latex] (s2) -- (d1);
		\draw[-latex] (s2) -- (v3);
		\draw[-latex] (s2) -- (d4);
		\draw[-latex] (s2) -- (d5);
		
		\draw[-latex] (s4) -- (d1);
		\draw[-latex] (s4) -- (d2);
		\draw[-latex] (s4) -- (v3);
		\draw[-latex] (s4) .. controls(2,3) .. (d5);

		\draw[-latex] (s5) -- (d1);
		\draw[-latex] (s5) -- (d2);
		\draw[-latex] (s5) -- (v3);
		\draw[-latex] (s5) .. controls(2,3) .. (d4);

		\draw[-latex] (v3) -- (d1);
		\draw[-latex] (v3) -- (d2);
		\draw[-latex] (v3) -- (d4);
		\draw[-latex] (v3) -- (d5);
	\end{tikzpicture}}
\caption{Example of network sharing} \label{fig:sharing}
\end{figure}
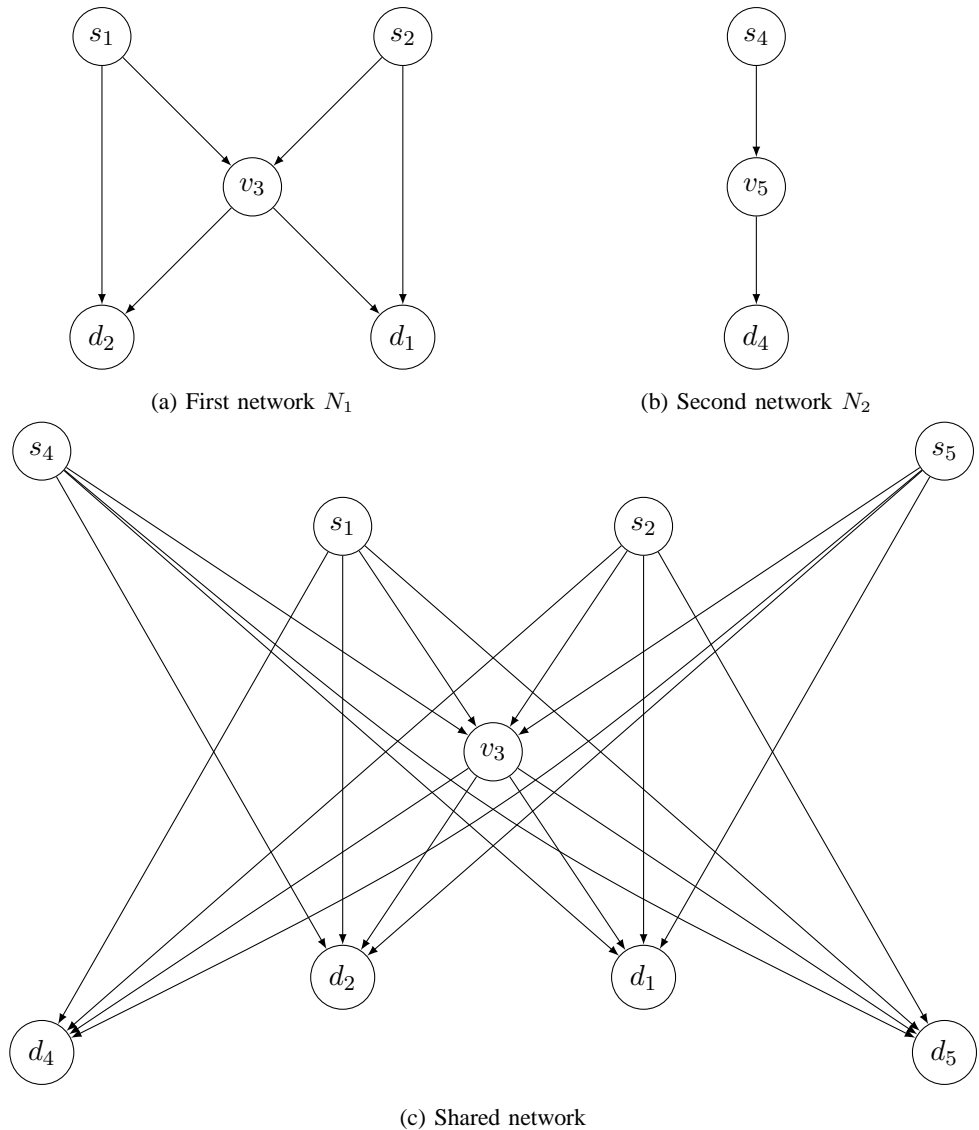

\subsection{Combining alphabets}

Let $[k] = \{1,\ldots,k\}$ for any positive integer $k$. We define a closure operator on $V \times [k]$ as follows. For any $v \in V$, let $[v] = \{(v,i): i \in [k]\}$ and for any $X \subseteq V \times [k]$, denote $X_V = \{v \in V: [v] \subseteq X\}$. Then
$$
	\cl^{[k]}(X) := X \cup \{[v] : v \in \cl(X_V)\}.
$$
This closure operator can be intuitively explained as follows. Consider the solvability problem of $\cl$ over the alphabet $A^k$. Each element of $A^k$ is a vector of length $k$ over $A$, then $\cl^{[k]}$ associates $k$ according vertices $[v]$ to each  $v \in V$, each new vertex $(v,i)$ corresponding to the coordinate $i$. If $v \in \cl(Y)$ for some $Y \subseteq V$, then the local function $f_v$ depends on $f_Y$. We can view $f_v: A^{kr} \to A^k$ (and hence all its coordinate functions) as depending on all coordinates of all vertices in $Y$, hence the definition of the closure operator.

In particular, for $D$ construct $D^{[k]}$ as follows: its vertex set is $V \times [k]$ and its edge set is $\{((u,i), (v,j)) : (u,v) \in E(D)\}$. Then it is easy to check that
$\cl_D^{[k]} = \cl_{D^{[k]}}.$

\begin{proposition} \label{prop:clk}
We have the following properties:
\begin{enumerate}
	\item $r(\cl^{[k]}) = k r(\cl)$.

    \item $G(\cl^{[k]}, A) \cong G(\cl,A^k)$ and hence $H(\cl^{[k]}, A) = kH(\cl,A^k)$.


    \item If $\cl$ is connected, then so is $\cl^{[k]}$ for all $k$.
\end{enumerate}
\end{proposition}

\begin{IEEEproof}
The proof of the first two claims is similar to that of \cite[Proposition 10]{GR11}. We now prove the last claim. For any $S \subseteq V \times [k]$, we denote $\lfloor S \rfloor = \bigcup_{v \in S_V} [v]$, $T = (V \times [k]) \backslash S$, and $\lceil T \rceil = T \cup (S \backslash \lfloor S \rfloor) = (V \times [k]) \backslash \lfloor S \rfloor$. Note that $S_V = \lfloor S \rfloor_V = V \backslash \lceil T \rceil_V$. Then we claim that if $\cl^{[k]}|_{\backslash T} = \cl^{[k]}|^{/ T}$, then $\cl^{[k]}|_{\backslash \lceil T \rceil} = \cl^{[k]}|^{/ \lceil T \rceil}$. For any $Y \subseteq \lfloor S \rfloor$, let $X = Y \cup (S \backslash \lfloor S \rfloor)$; then $X_V = Y_V$ and $X \cup T = Y \cup \lceil T \rceil$. We then have
\begin{align*}
	\{[v] : v \in \cl(Y_V)\} \cap S &= \{[v]: v \in \cl(X_V)\} \cap S\\
	&= \{[v]: v \in \cl((X \cup T)_V)\} \cap S\\
	&= \{[v]: v \in \cl((Y \cup T)_V)\} \cap S,
\end{align*}
and in particular, then intersections with $\lfloor S \rfloor$ are equal, thus proving the claim.

Now suppose $\cl^{[k]}$ is disconnected, then $\cl^{[k]}|_{\backslash T} = \cl^{[k]}|^{/ T}$ for some $T = \lceil T \rceil$ (and hence $S = \lfloor S \rfloor$ and $V = S_V \cup T_V$). Then for any $X \subseteq S$, $(X \cup T)_V = X_V \cup T_V$ and we have
\begin{align*}
	\{[v] : v \in \cl(X_V) \cap S_V\}  &= \{[v] : v \in \cl(X_V)\} \cap S\\
	&= \{[v] : v \in \cl(X_V \cup T_V)\} \cap S\\
	&= \{[v] : v \in \cl(X_V \cup T_V) \cap S_V\},
\end{align*}
and hence $\cl|_{\backslash T_V}(X_V) = \cl|^{/ T_V}(X_V)$ for all $X_V \subseteq S_V$.
\end{IEEEproof}

\section{Acknowledgment} \label{sec:ack}

The author would like to thank the team from Queen Mary, University of London for fruitful discussions and the anonymous reviewers for their valuable suggestions.

\appendix

\subsection{Proof of Theorem \ref{th:algorithm}} \label{app:th}

First of all, we justify why we only search for useless vertices in $T(D)$.

\begin{lemma} \label{lem:useless_T}
If $V_2$ is useless, then $V_2 \subseteq T(D)$.
\end{lemma}

\begin{IEEEproof}
Let $V_2$ be a useless set. First of all, if $X$ induces a chordless cycle, then it cannot entirely lie in $V_2$, for $V_2$ is acyclic. Suppose $X$ does not lie entirely in $V_1$ either. Since $X_1$ is acyclic, we have $X_1 \subseteq \cl_D/_{V_2}(Y) \subseteq \cl_D(Y)$, where $Y = V_1 \backslash X_1$. Therefore, $X_2 \subseteq X^- \subseteq \cl_D(Y)$; gathering, we obtain $X \subseteq \cl_D(Y)$. More precisely, $X \subseteq \cl_D(Y) \backslash Y$ and hence $X$ is acyclic, which is a contradiction.
\end{IEEEproof}

The following results ensure that we can remove useless vertices one by one.

\begin{lemma} \label{lem:v_removed}
Let $V_2$ be useless in $D$ and $v \in V_2$. Once $v$ is removed from $D$, $V_2 \backslash v$ is useless in $D[V \backslash v]$.
\end{lemma}

\begin{IEEEproof}
$V_2 \backslash v$ is clearly acyclic. For any $X \subseteq (V \backslash v)$, we have
\begin{align*}
	\cl_{D[V\backslash v]}(X) \backslash (V_2 \backslash v) &= \cl_D(X \cup v) \backslash V_2\\
	&= \cl_D(X \cup V_2) \backslash V_2\\
	&= (\cl_D(X \cup V_2) \backslash v) \backslash (V_2 \backslash v)\\
	&= \cl_{D[V \backslash v]}(X \cup (V_2 \backslash v)) \backslash (V_2 \backslash v).
\end{align*}
\end{IEEEproof}

\begin{lemma} \label{lem:topological_useless}
Let $V_2$ be useless in $D$ and $v$ the last vertex of $V_2$ according to the topological order (i.e., $v^+ \subseteq S$). Then $\{v\}$ is a useless set.
\end{lemma}

\begin{IEEEproof}
We only need to prove that $\{v\}$ is weak, i.e. for all $X \subseteq V \backslash v$, $\cl_D(X \cup v) \backslash v = \cl_D(X) \backslash v$. This clearly holds if $v \in \cl_D(X)$, hence let us assume that $v \notin \cl_D(X)$.

It is easy to show by induction on $j$ that $c_D^j(X) \backslash X = c_D^j(X \cup v) \backslash (X \cup v)$ if and only if $c_D^j(X \cup v) \backslash (X \cup v) \cap v^+ = \emptyset$; in particular, $\cl_D(X) \backslash X = \cl_D(X \cup v) \backslash (X \cup v)$ if and only if $(\cl_D(X \cup v) \backslash (X \cup v)) \cap v^+ = \emptyset$.

We have $v^+ \cap (\cl_D(X) \backslash X) = \emptyset$. Since $V_2$ is weak, $\cl_D(X) \cap V_1 = \cl_D(X \cup v) \cap V_1$. Moreover, since $v^+ \subseteq V_1$, we have $\cl_D(X) \cap v^+ = \cl_D(X \cup v) \cap v^+$. Combining, we obtain $(\cl_D(X \cup v) \backslash (X \cup v)) \cap v^+ = \emptyset$ and by the paragraph above, $\cl_D(X) \backslash X = \cl_D(X \cup v) \backslash (X \cup v)$, which yields $\cl_D(X) \backslash v = \cl_D(X \cup v) \backslash v$.
\end{IEEEproof}

Next, we indicate an efficient way to check that a singleton is useless.

\begin{lemma} \label{lem:v_useless}
For any vertex $v$, $\{v\}$ is useless if and only if for any $u \in v^+$, $v \in \cl_D(u^- \backslash v)$.
\end{lemma}

\begin{IEEEproof}
Suppose there exists $u \in v^+$ such that $v \notin \cl_D(u^- \backslash v)$. There is an edge from $v$ to $u$, hence $u \notin \cl_D(u^- \backslash v) \backslash (u^- \backslash v)$ and $u \notin \cl_D(u^- \backslash v) \backslash u^-$. Since $u \in \cl_D(u^-) \backslash u^-$, we obtain $\cl_D(u^-) \backslash v \neq \cl_D(u^- \backslash v) \backslash v$ and $\{v\}$ is not weak.

Otherwise, suppose there exists $X$ such that $\cl_D(X) \backslash v \neq \cl_D(X \cup v) \backslash v$; clearly $v \notin \cl_D(X)$. It is easy to show by induction on $j$ that $c_D^j(X) \backslash X = c_D^j(X \cup v) \backslash (X \cup v)$ if and only if $\big(c_D^j(X \cup v) \backslash (X \cup v)\big) \cap v^+ = \emptyset$. Let $i = \min\{j : c_D^j(X) \backslash X \neq c_D^j(X \cup v) \backslash (X \cup v)\}$, then there exists $u \in (c_D^i(X \cup v) \backslash c_D^{i-1}(X \cup v)) \cap v^+$. We have $u^- \backslash v \subseteq c_D^{i-1}(X) \cl_D(X)$, and hence $v \in \cl_D(u^- \backslash v) \subseteq \cl_D(X)$, which is the desired contradiction.
\end{IEEEproof}

We can now prove the correctness of Algorithm \ref{alg:1}. Clearly, the running time is polynomial.

\begin{IEEEproof}
First of all, Lemma \ref{lem:useless_T} guarantees that the set of useless vertices lies in $T(D)$. At every iteration of the Repeat loop, if there exists a set of useless vertices in the new graph, then there exists a singleton $\{v\}$ which is useless by \ref{lem:topological_useless}. By Lemma \ref{lem:v_useless}, the algorithm will find a useless vertex $v$ if there exists one. Lemma \ref{lem:v_removed} guarantees that after all the iterations, all the useless vertices will be removed.
\end{IEEEproof}

\bibliographystyle{IEEEtran}
\bibliography{g}

\begin{thebibliography}{10}
\providecommand{\url}[1]{#1}
\csname url@samestyle\endcsname
\providecommand{\newblock}{\relax}
\providecommand{\bibinfo}[2]{#2}
\providecommand{\BIBentrySTDinterwordspacing}{\spaceskip=0pt\relax}
\providecommand{\BIBentryALTinterwordstretchfactor}{4}
\providecommand{\BIBentryALTinterwordspacing}{\spaceskip=\fontdimen2\font plus
\BIBentryALTinterwordstretchfactor\fontdimen3\font minus
  \fontdimen4\font\relax}
\providecommand{\BIBforeignlanguage}[2]{{%
\expandafter\ifx\csname l@#1\endcsname\relax
\typeout{** WARNING: IEEEtran.bst: No hyphenation pattern has been}%
\typeout{** loaded for the language `#1'. Using the pattern for}%
\typeout{** the default language instead.}%
\else
\language=\csname l@#1\endcsname
\fi
#2}}
\providecommand{\BIBdecl}{\relax}
\BIBdecl

\bibitem{ACLY00}
R.~Ahlswede, N.~Cai, S.-Y.~R. Li, and R.~W. Yeung, ``Network information
  flow,'' \emph{IEEE Transactions on Information Theory}, vol.~46, no.~4, pp.
  1204--1216, July 2000.

\bibitem{LYC03}
S.-Y.~R. Li, R.~W. Yeung, and N.~Cai, ``Linear network coding,'' \emph{IEEE
  Transactions on Information Theory}, vol.~49, no.~2, pp. 371--381, February
  2003.

\bibitem{Rii04}
S.~Riis, ``Linear versus non-linear boolean functions in network flow,'' in
  \emph{Proc. CISS}, Princeton, NJ, March 2004.

\bibitem{DFZ05}
R.~Dougherty, C.~Freiling, and K.~Zeger, ``Insufficiency of linear coding in
  network information flow,'' \emph{IEEE Transactions on Information Theory},
  vol.~51, no.~8, pp. 2745--2759, August 2005.

\bibitem{KM03}
R.~K\"otter and M.~M\'edard, ``An algebraic approach to network coding,''
  \emph{IEEE/ACM Transactions on Networking}, vol.~11, no.~5, pp. 782--795,
  October 2003.

\bibitem{HMK+06}
T.~Ho, M.~M\'edard, R.~K\"otter, D.~R. Karger, M.~Effros, J.~Shi, and B.~Leong,
  ``A random linear network coding approach to multicast,'' \emph{IEEE
  Transactions on Information Theory}, vol.~52, no.~10, pp. 4413--4430, October
  2006.

\bibitem{DFZ04}
R.~Dougherty, C.~Freiling, and K.~Zeger, ``Linearity and solvability in
  multicast networks,'' \emph{IEEE Transactions on Information Theory},
  vol.~50, no.~10, pp. 2243--2256, October 2004.

\bibitem{DFZ06}
------, ``Unachievability of network coding capacity,'' \emph{IEEE Transactions
  on Information Theory}, vol.~52, no.~6, pp. 2365--2372, June 2006.

\bibitem{DFZ07}
------, ``Networks, matroids, and non-{S}hannon information inequalities,''
  \emph{IEEE Transactions on Information Theory}, vol.~53, no.~6, pp.
  1949--1969, June 2007.

\bibitem{CG08}
T.~Chan and A.~Grant, ``Dualities between entropy functions and network
  codes,'' \emph{IEEE Transactions on Information Theory}, vol.~54, no.~10, pp.
  4470--4487, October 2008.

\bibitem{CG10}
------, ``On capacity regions of non-multicast networks,'' in \emph{Proc.
  ISIT}, Austin, TX, 2010, pp. 2378--2382.

\bibitem{Cha05}
T.~H. Chan, ``On the optimality of group network codes,'' in \emph{Proc. ISIT},
  Adelaide, Australia, 2005, pp. 1992--1996.

\bibitem{Rii06}
S.~Riis, ``Utilising public information in network coding,'' in \emph{General
  Theory of Information Transfer and Combinatorics}, 2006, pp. 866--897.

\bibitem{DZ06}
R.~Dougherty and K.~Zeger, ``Nonreversibility and equivalent constructions of
  multiple unicast networks,'' \emph{IEEE Transactions on Information Theory},
  vol.~52, no.~11, pp. 1287--1291, November 2006.

\bibitem{Rii07}
S.~Riis, ``Information flows, graphs and their guessing numbers,'' \emph{The
  Electronic Journal of Combinatorics}, vol.~14, pp. 1--17, 2007.

\bibitem{WCR09}
T.~Wu, P.~Cameron, and S.~Riis, ``On the guessing number of shift graphs,''
  \emph{Journal of Discrete Algorithms}, vol.~7, pp. 220--226, 2009.

\bibitem{Rii07b}
S.~Riis, ``Reversible and irreversible information networks,'' \emph{IEEE
  Transactions on Information Theory}, vol.~53, no.~11, pp. 4339--4349,
  November 2007.

\bibitem{Cha11}
T.~Chan, ``Recent progresses in characterising information inequalities,''
  \emph{Entropy}, vol.~13, pp. 379--401, 2011.

\bibitem{CM11}
D.~Christofides and K.~Markstr\"{o}m, ``The guessing number of undirected
  graphs,'' \emph{Electronic Journal of Combinatorics}, vol.~18, no.~1, pp.
  1--19, 2011.

\bibitem{GR11}
M.~Gadouleau and S.~Riis, ``Graph-theoretical constructions for graph entropy
  and network coding based communications,'' \emph{IEEE Transactions on
  Information Theory}, vol.~57, no.~10, pp. 6703--6717, October 2011.

\bibitem{Sha79}
A.~Shamir, ``How to share a secret,'' \emph{Communications of the ACM},
  vol.~22, no.~11, p. 612–613, November 1979.

\bibitem{BD91}
E.~F. Brickell and D.~M. Davenport, ``On the classification of ideal secret
  sharing schemes,'' \emph{J. Cryptology}, vol.~4, pp. 123--134, 1991.

\bibitem{Sey92}
P.~D. Seymour, ``On secret-sharing matroids,'' \emph{J. Combin. Theory Ser. B},
  vol.~56, pp. 69--73, 1992.

\bibitem{Mat99}
F.~Mat\'u\v{s}, ``Matroid representations by partitions,'' \emph{Discrete
  Math.}, vol. 203, pp. 169--194, 1999.

\bibitem{CGR13}
P.~J. Cameron, M.~Gadouleau, and S.~Riis, ``Combinatorial representations,''
  \emph{Journal of Combinatorial Theory, Series A}, vol. 120, no.~3, pp.
  671--682, April 2013.

\bibitem{WS10}
C.-C. Wang and N.~B. Shroff, ``Pairwise intersession network coding on directed
  networks,'' \emph{IEEE Transactions on Information Theory}, vol.~56, no.~8,
  pp. 38\,793--3900, August 2010.

\bibitem{Bir48}
G.~Birkhoff, \emph{Lattice theory}.\hskip 1em plus 0.5em minus 0.4em\relax
  American Mathematical Society Colloquium Publications, 1948.

\bibitem{Oxl06}
J.~G. Oxley, \emph{Matroid Theory}.\hskip 1em plus 0.5em minus 0.4em\relax
  Oxford University Press, 2006.

\bibitem{Bai04}
R.~A. Bailey, \emph{Association Schemes: Designed Experiments, Algebra and
  Combinatorics}, ser. Cambridge Series in Statistical and Probabilistic
  Mathematics.\hskip 1em plus 0.5em minus 0.4em\relax Cambridge: Cambridge
  University Press, 2004.

\bibitem{BG09a}
J.~Bang-Jensen and G.~Gutin, \emph{Digraphs: Theory, Algorithms and
  Applications}.\hskip 1em plus 0.5em minus 0.4em\relax Springer, 2009.

\bibitem{GR01}
C.~D. Godsil and G.~Royle, \emph{Algebraic Graph Theory}, ser. Graduate Texts
  in Mathematics.\hskip 1em plus 0.5em minus 0.4em\relax Springer-Verlag, 2001,
  vol. 207.

\bibitem{Lov75}
L.~Lov\'asz, ``On the ratio of optimal integral and fractional covers,''
  \emph{Discrete Mathematics}, vol.~13, pp. 383--390, 1975.

\bibitem{BSP60}
R.~Bose, S.~Shrikhande, and E.~Parker, ``Further results on the construction of
  mutually orthogonal latin squares and the falsity of {E}uler's conjecture,''
  \emph{Canadian Journal of Mathematics}, vol.~12, pp. 189--203, 1960.

\end{thebibliography}

\end{document}